\documentclass{journalA4}

\usepackage{amssymb,amsmath,xspace,graphicx,wrapfig}
\usepackage[usenames]{color}
\usepackage{mdwlist}
\usepackage{url,hyperref}
\usepackage[numbers]{natbib}
\usepackage{times}

\makeatletter
\renewcommand*{\@fnsymbol}[1]{\ensuremath{\ifcase#1\or *\or \mathsection\or \mathparagraph\or
   \dagger\or \ddagger\or \|\or **\or \dagger\dagger
   \or \ddagger\ddagger \else\@ctrerr\fi}}
\makeatother

\graphicspath{{figures/}}

\definecolor {infocolor} {rgb} {0.6,0.6,0.6}

\newcommand{\eps}{\varepsilon}

\newcommand{\comp}[1]{\ensuremath{C_{#1}}\xspace}
\newcommand{\groups}[1]{\ensuremath{\G_{#1}}\xspace}

\newcommand{\mkmcal}[1]{\ensuremath{\mathcal{#1}}\xspace}

\newcommand{\G}{\mkmcal{G}}

\newcommand{\T}{\mkmcal{T}}
\newcommand{\C}{\mkmcal{C}}
\newcommand{\X}{\mkmcal{X}}

\newcommand{\M}{\mkmcal{M}}
\newcommand{\RR}{\mkmcal{R}}

\newcommand{\problem}{\textsc}

\newcommand{\mkmbb}[1]{\ensuremath{\mathbb{#1}}\xspace}

\newcommand{\R}{\mkmbb{R}}

\renewcommand{\paragraph}[1]{\smallskip\noindent\textbf{#1}.}

\hypersetup{
colorlinks,%
citecolor=black,%
filecolor=black,%
linkcolor=black,%
urlcolor=black
}

\title{Trajectory Grouping Structure}
\author{
Kevin Buchin\footnote{Dep. of Mathematics and Computer Science, TU Eindhoven,
  The Netherlands, {\tt \{k.a.buchin, m.e.buchin\}@tue.nl} and {\tt speckman@win.tue.nl}.
MB and BS are supported by the Netherlands' Organisation for Scientific
Research (NWO) under project no. 612.001.106. and 639.022.707,
respectively.}
\and
Maike Buchin\footnotemark[1]
\and
Marc van Kreveld\footnote{Dep. of Information and Computing Sciences, Utrecht
  University, The Netherlands, {\tt m.j.vankreveld@uu.nl} and {\tt f.staals@uu.nl}.
FS is supported by the Netherlands' Organisation for Scientific Research (NWO) under project no. 612.001.022.}
\and
Bettina Speckmann\footnotemark[1]
\and
Frank Staals\footnotemark[2]
}

\date{}

\begin{document}

\maketitle
\begin{abstract}\noindent
The collective motion of a set of moving entities like people,
birds, or other animals, is characterized by groups arising,
merging, splitting, and ending. Given the trajectories of these
entities, we define and model a structure that captures all of
such changes using the Reeb graph, a concept from topology.
The \emph{trajectory grouping structure} has three natural parameters
that allow more global views of the data in group size, group duration, and entity
inter-distance. We prove complexity bounds on the maximum number
of maximal groups that can be present, and give algorithms to
compute the grouping structure efficiently. We also study how
the trajectory grouping structure can be made robust, that is,
how brief interruptions of groups can be disregarded in the
global structure, adding a notion of persistence to the structure.
Furthermore, we showcase the results of
experiments using data generated by the NetLogo flocking model and
from the Starkey project. The Starkey data describe the movement of
elk, deer, and cattle. Although there is no ground truth for the 
grouping structure in this data, the
experiments show that the trajectory grouping structure is
plausible and has the desired effects when changing the essential
parameters.
Our research provides the first complete study of trajectory group
evolvement, including combinatorial, algorithmic, and experimental results.

\end{abstract}




\thispagestyle{empty}
\setcounter{page}{0}
\clearpage

\section{Introduction}
\label{sec:Introduction}

In recent years there has been an increase in location-aware devices and
wireless communication networks. This has led to a large amount of trajectory
data capturing the movement of animals, vehicles, and people. The increase in
trajectory data goes hand in hand with an increasing demand for techniques and
tools to analyze them, for example, in transportation sciences, sports,
ecology, and social services.

An important task is the analysis of movement patterns. In particular, given a 
set of moving entities 
we wish to determine when and which subsets of entities travel together.
When a sufficiently large set of entities travels
together for a sufficiently long time, we call such a set a \emph{group} (we
give a more formal definition later). Groups may start, end, split and merge
with other groups. Apart from the question what the current groups are, we
also want to know which splits and merges led to the current groups, when they
happened, and which groups they involved. We wish to capture this group change
information in a model that we call the \emph{trajectory grouping structure}.

The informal definition above suggests that three parameters are needed to
define groups: (i) a spatial parameter for the distance between entities;
(ii) a temporal parameter for the duration of a group; (iii) a count for
the number of entities in a group.
We will design our grouping structure definition
to incorporate these parameters so that we can study grouping at different
scales. We use the three parameters as follows: a small spatial parameter implies
we are interested only in spatially close groups, a large temporal parameter implies
we are interested only in long-lasting groups, and a large count implies we
are interested only in large groups. By adjusting the parameters suitably,
we can obtain more detailed or more generalized views of the trajectory grouping structure.

The use of scale parameters and the fact that the grouping
structure changes at discrete events suggest the use of
computational topology~\cite{eh-cti-10}.
In particular, we use Reeb graphs to capture the grouping structure.
Reeb graphs have been used extensively in shape analysis and the visualization
of scientific data (see e.g.~\cite{Biasotti2008,Edelsbrunner2008,Fomenko1997}).
A Reeb graph captures the structure of a two- or higher-dimensional
scalar function, by considering the evolution of the connected components
of the level sets.
The computation of Reeb graphs has received considerable attention in
computational geometry and topology; an overview is given in
\cite{DeyWang2011}.
Recently, a deterministic $O(n\log n)$ time algorithm
was presented for constructing the Reeb graph of a 2-skeleton of size
$n$~\cite{parsa2012connectivity}.
Edelsbrunner et al.~\cite{Edelsbrunner2008} discuss time-varying Reeb graphs
for continuous space-time data. Although we also analyze continuous
space-time data (2D-space in our case),
our Reeb graphs are not time-varying, but
time is the parameter that defines the Reeb graph.
Ge et al.~\cite{GeSBW11} use the Reeb graph to compute a one-dimensional ``skeleton'' from unorganized data. 
In contrast to our setting, in their applications the data comes without a time component. 
They use a proximity graph on the input points to build a simplicial complex from which they compute the Reeb graph.

Our research is motivated by and related to previous research on
flocks~\cite{benkert2008reporting, gudmundsson2006computing,gudmundsson2007efficient, vieira2009onlineflock},
herds~\cite{huang2008modeling},
convoys~\cite{jeung2008discovery},
moving clusters~\cite{kalnis2005moving},
mobile groups~\cite{hwang2005mobilegroup,Wang2008mobilegroup}
and swarms~\cite{li2010swarm}.
These concepts differ from each other in the way in which space and time
are used to test if entities form a group: do the entities stay in a single
disc or are they density-connected~\cite{ester1996density},
should they stay together during consecutive time steps or not, can the
group members change over time, etc.
Only the herds concept~\cite{huang2008modeling} includes the splitting
and merging of groups.

\paragraph{Contributions}
We present the first complete study of trajectory group
evolvement, including combinatorial, algorithmic, and experimental results.
Our research differs from and improves on previous research in the following ways:
Firstly, our model is simpler than herds and thus more intuitive.
Secondly, we consider the grouping structure at continuous times instead
of at discrete steps (which was done only for flocks).
Thirdly, we analyze the algorithmic and combinatorial aspects
of groups and their changes.
Fourthly, we implemented our algorithms and provide evidence that our model
captures the grouping structure well and can be computed efficiently.
Fifthly, we extend the model to incorporate persistence.

We created videos based on our implementation showing the maximal groups we
found in simulated NetLogo flocking data~\cite{netlogo_flocking,netlogo} and
in real-world data from the Starkey project~\cite{starkey}.


\clearpage

\paragraph{A Definition for a Group}
Let \X be a set of entities of which we have locations during some time span.
The $\eps$\emph{-disc} of an entity $x$ (at time $t$) is a disc of radius $\eps$
centered at $x$ at time $t$.
Two entities are \emph{directly connected}
at time $t$ if their $\eps$-discs overlap.
Two entities $x$ and $y$ are $\eps$\emph{-connected} at
time $t$ if there is a sequence $x=x_0,..,x_k=y$ of entities such that for all
$i$, $x_i$ and $x_{i+1}$ are directly connected.

A subset $S \subseteq \X$ of entities is $\eps$-connected at time $t$ if all
entities in $S$ are pairwise $\eps$-connected at time $t$. This means that
the union of the $\eps$-discs of entities in $S$ forms a single connected region.
The set $S$ forms a \emph{component} at time $t$ if and only if $S$ is
$\eps$-connected, and $S$ is maximal with respect to this property. The set
of components $\C(t)$ at time $t$ forms a partition of the entities in \X at time $t$.

Let the spatial parameter of a group be $\eps$, the temporal parameter $\delta$,
and the size parameter $m$.
A set $G$ of $k$ entities forms a \emph{group} during time interval $I$
if and only if the following three conditions hold:
(\textit{i}) $G$ contains at least $m$ entities, so $k \geq m$, (\textit{ii})
the interval $I$ has length at least $\delta$, and (\textit{iii}) at all times
$t \in I$, there is a component $C \in \C(t)$ such that $G \subseteq C$.

\begin{wrapfigure}[9]{r}{0.4\textwidth}
  \centering
  \vspace{-1.5\baselineskip}
  \includegraphics{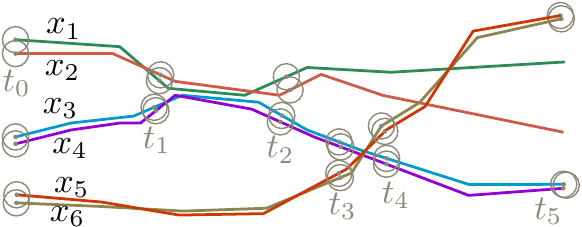}
  \caption{For $m=2$ and $\delta > t_4 - t_3$ there are four maximal groups:
    $\{x_1,x_2\}$, $\{x_3,x_4\}$, $\{x_5,x_6\}$, and $\{x_1,..,x_4\}$.
   }
  \label{fig:example_groups}
\end{wrapfigure}

We denote the interval $I = [t_s,t_e]$ of group $G$ with $I_G$.  Group $H$
\emph{covers} group $G$ if $G\subseteq H$ and $I_G\subseteq I_H$.  If there are
no groups that cover $G$, we say $G$ is \emph{maximal} (on $I_G$). In
Fig.~\ref{fig:example_groups}, groups $\{x_1,x_2\}$, $\tilde{G} = \{x_3,x_4\}$,
$\hat{G} = \{x_5,x_6\}$, and $G$ $=\{x_1,..,x_4\}$ are maximal: $\tilde{G}$ and
$\hat{G}$ on $[t_0,t_5]$, $G$ on $[t_1,t_2]$. Group $\{x_1,x_3\}$ is covered by $G$
and hence not maximal.

Note that entities can be in multiple maximal groups at the same time.
For example, entities $\{y_1,y_2,y_3\}$ can travel together for a while,
then $y_4,y_5$ may become $\eps$-connected, and shortly thereafter $y_1,y_4,y_5$
separate and travel together for a while. Then $y_1$ may be in two otherwise
disjoint maximal groups for a short time. An entity can also be in two maximal
groups where one is a subset of the other. In that case the group with fewer
entities must last longer. That an entity is in more groups simultaneously
may seem counterintuitive at first, but it is necessary to capture all
grouping information.
We will show that the total number of maximal groups is $O(\tau n^3)$, where
$n$ is the number of entities in \X and $\tau$ is the number of edges of each input trajectory. This bound is tight in the worst case.

Our maximal group definition uses three parameters, which all allow a
more global view of the grouping structure. In particular, we observe
that there is \emph{monotonicity} in the group size and the duration:
If $G$ is a group during interval $I$,
and we decrease the minimum required group size $m$ or
decrease the minimum required duration $\delta$,
then $G$ is still a group on time interval $I$. Also, if $G$ is a maximal
group on $I$, then it is also a maximal group for a smaller $m$ or smaller
$\delta$.
For the spatial parameter $\eps$ we observe monotonicity in a slightly
different manner: if $G$ is a group for a given $\eps$, then for a larger
value of $\eps$ there exists a group $G'\supseteq G$.
The monotonicity property is important when we want to have a more detailed
view of the data: we do not lose maximal groups in a more detailed view.
The group may, however, be extended in size and/or duration.

We capture the grouping structure using a Reeb graph of the $\eps$-connected
components together with the set of all maximal groups. Parts of the Reeb graph
that do not support a maximal group can be omitted.
The grouping structure can help us in answering various questions. For example:
\begin{itemize*}
\item What is the largest/longest maximal group at time $t$?
\item How many entities are currently (not) in any maximal group?
\item What is the first maximal group that starts/ends after time $t$?
\item What is the total time that an entity was part of any maximal group?
\item Which entity has shared maximal groups with the most other entities?
\end{itemize*}
Furthermore, the grouping structure can be used to partition the trajectories
in independent data sets, to visualize grouping aspects of the trajectories, and
to compare grouping across different data sets.


We also discuss 
robustness of the grouping structure
in the following sense.
If an entity $x$ leaves a group $G$ and almost immediately returns, we would
like to ignore the small interval on which $x$ and $G$ were separate,
and just consider $G \cup \{x\}$ as one group.
The maximal group definition given above is not robust, 
but later in the paper we will study an extension that is. 
Note that robustness requires an additional parameter that captures how
short any interruption in a group may last to be ignored.


\paragraph{Results and Organization}
We discuss how to represent the grouping structure in
Section~\ref{sec:Representing_the_grouping_structure}, and prove
that there are always $O(\tau n^3)$ maximal groups, which is tight in the worst case.
Here $n$ is the number
of trajectories (entities) and $\tau$ the number of edges in each
trajectory. We present an algorithm to compute the trajectory grouping
structure and all maximal groups in
Section~\ref{sec:Computing_the_Grouping_Structure}. This algorithm runs in
$O(\tau n^3+ N)$ time, where $N$ is the total output size.
In Section~\ref{sec:Robustness} we make our definitions more robust,
and extend our algorithms to this case.
In Section~\ref{sec:Evaluation} we evaluate our methods on synthetic and real-world data.

\begin{figure*}[b]
  \centering
  \includegraphics{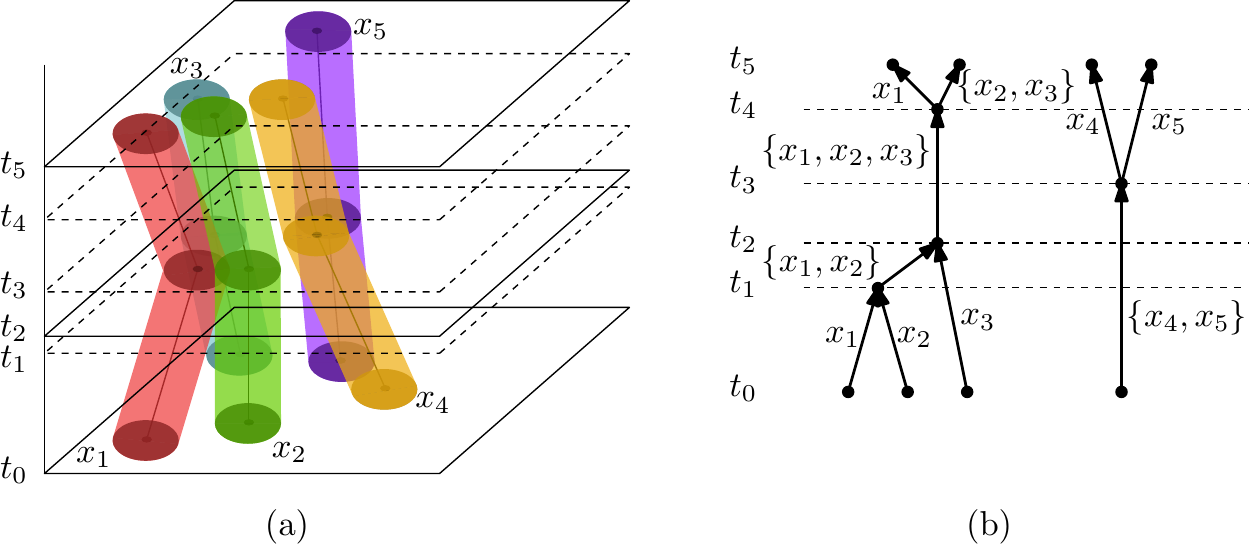}
  \caption{The manifold for the entities $\X = \{ x_1,..,x_5 \}$ (a), and the
    corresponding Reeb graph (b).}
  \label{fig:Reeb_graph}
\end{figure*}

\section{Representing the Grouping Structure}
\label{sec:Representing_the_grouping_structure}

Let \X be a set of $n$ entities, where each entity travels along a path of
$\tau$ edges. To compute the grouping structure we consider a manifold \M in
$\R^3$, where the $z$-axis corresponds to time. The manifold \M is the union of $n$
``tubes'' (see Fig.~\ref{fig:Reeb_graph}(a)). Each tube consists of $\tau$ skewed
cylinders with horizontal radius $\eps$ that we obtain by tracing the
$\eps$-disc of an entity $x$ over its trajectory.

Let $H_t$ denote the
horizontal plane at height $t$, then the set $\M \cap H_t$ is the \emph{level
  set} of $t$. The connected components in the level set of $t$ correspond to
the components (maximal sets of $\eps$-connected entities) at time $t$.  We
will assume for simplicity that all trajectories have their known positions at
the same times $t_0,..,t_\tau$ and that no three entities become
$\eps$-(dis)connected at the same time, but most of our theory does not depend
on these assumptions.

\subsection{The Reeb Graph}

We start out with a possibly disconnected solid that is the union of a collection of
tube-like regions: a 3-manifold with boundary.
Note that this manifold is not explicitly defined.
We are interested in horizontal cross-sections, and the evolution of the
connected components of these cross-sections defines the Reeb graph.
Note that this is different from the usual Reeb graph that is obtained
from the 2-manifold that is the boundary of our 3-manifold, using the
level sets of the height function (the function whose level sets we follow
is the height function above a horizontal plane below the manifold), see~\cite{eh-cti-10} for a background on these topics.

To describe how the components change over time, we consider
the Reeb graph \RR of \M (Fig.~\ref{fig:Reeb_graph}(b)).
The Reeb graph has a vertex $v$ at every time $t_v$ where the components change.
The vertex times are usually not at any of the given times $t_0,..,t_\tau$, but in between
two consecutive time steps.
The vertices of the Reeb graph can be classified in four groups.
There is a \emph{start vertex} for
every component at $t_0$ and an \emph{end vertex} at $t_\tau$. A start vertex
has in-degree zero and out-degree one, and an end vertex has in-degree one and
out-degree zero. The remaining vertices are either \emph{merge vertices} or
\emph{split vertices}. Since we assume that no three entities become
$\eps$-(dis)connected at exactly the same time there are no simultaneous splits and
merges. This means merge vertices have in-degree two and out-degree one,
and split vertices have in-degree one and out-degree two. A directed edge
$e=(u,v)$ connecting vertices $u$ and $v$, with $t_u < t_v$, corresponds to a
set $C_e$ of entities that form a component at any time $t \in I_e =
[t_u,t_v]$. The Reeb graph is this directed graph. Note that the Reeb graph
depends on the spatial parameter $\eps$, but not on the other two parameters of
maximal groups.

\begin{lemma}
  \label{lem:lowerbound_size_Reeb_graph}
  The Reeb graph $\RR$ for a set \X of $n$ entities, each of which travels along a
  trajectory of $\tau$ edges, can have $\Omega(\tau n^2)$ vertices and
  $\Omega(\tau n^2)$ edges.
\end{lemma}

\begin{proof}
  We construct $n$ trajectory edges on which the entities travel in between two
  consecutive time stamps, say $t_i$ and $t_{i+1}$, such that the Reeb graph
  for $\eps = 0$ has $\Omega(n^2)$ vertices $v$ with $t_v \in
  [t_i,t_{i+1}]$. We use this construction in between all times $t_{2i}$ and
  $t_{2i+1}$, and move the entities back to their starting position in between
  $t_{2i+1}$ and $t_{2i+2}$. Therefore, the total number of vertices is
  $\Omega(\tau n^2)$. Since each vertex has degree one or three it follows that
  the number of edges is also $\Omega(\tau n^2)$.

  \begin{figure}[h]
    \centering
    \includegraphics{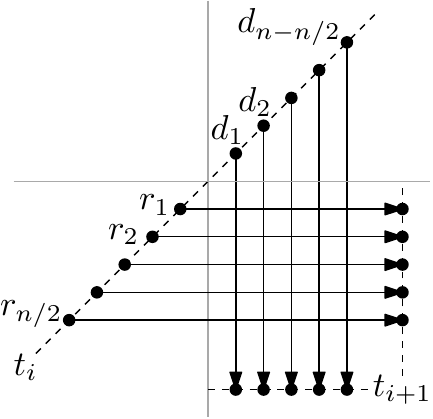}
    \caption{Every pair of entities $r_j$ and $d_\ell$ are at the same point at
      time $t_i + j + \ell$. This yields $\Omega(n^2)$ vertices in the interval
      $[t_i,t_{i+1}]$.}
    \label{fig:lowerbound_Reeb_graph}
  \end{figure}

  Let $\X = R \cup D$, with $R = r_1,..,r_{n/2}$ and $D = d_1,..,d_{n-n/2}$. At
  the start (time $t_i$) all entities start at the line $y=x$. In particular,
  we place $r_j$ on $(-j,-j)$ and $d_\ell$ on $(\ell,\ell)$. All entities move with
  speed one. The entities in $R$ move to the right, and the entities in $D$
  move downwards (see Fig.~\ref{fig:lowerbound_Reeb_graph}). It follows that
  each entity $r_j$ and $d_\ell$ are both at the same point at time $t_i + j +
  \ell$. Hence, we get a vertex in the Reeb graph. There are $\Omega(n^2)$ such
  intersections, and thus $\Omega(n^2)$ vertices. The lemma follows.
\end{proof}

\begin{theorem}
  \label{th:size_Reeb_graph}
  Given a set \X of $n$ entities, in which each entity travels along a
  trajectory of $\tau$ edges, the Reeb graph $\RR = (V,E)$ has $O(\tau n^2)$
  vertices and edges. These bounds are tight in the worst case.
\end{theorem}

\begin{proof}
  Lemma~\ref{lem:lowerbound_size_Reeb_graph} 
  gives
  a simple construction that shows that the Reeb graph may have $\Omega(\tau n^2)$
  vertices and edges. 
  For the upper bound,
  consider a trajectory edge $(v_i,v_{i+1})$ of (the trajectory of) entity $x
  \in \X$. An other entity $y \in \X$ is directly connected to $x$ during at
  most one interval $I \subseteq [t_i,t_{i+1}]$. This interval yields at most
  two vertices in \RR. The trajectory of $x$ consists of $\tau$ edges, hence a
  pair $x,y$ produces $O(\tau)$ vertices in \RR. This gives a total of
  $O(\tau n^2)$ vertices. Each vertex has constant degree, so there are
  $O(\tau n^2)$ edges. 
\end{proof}

\paragraph{The Trajectory Grouping Structure} The trajectories of entities are
associated with the edges of the Reeb graph in a natural way. Each entity
follows a directed path in the Reeb graph from a start vertex to an end
vertex. Similarly, (maximal) groups follow a directed path from a start or
merge vertex to a split or end vertex. If $m>0$ or $\delta>0$, there may be
edges in the Reeb graph with which no group is associated. These edges do not
contribute to the grouping structure, so we can discard them.  The remainder of
the Reeb graph we call the \emph{reduced Reeb graph}, which, together with all
maximal groups associated with its edges, forms the \emph{trajectory grouping
  structure}. 


\clearpage

\subsection{Bounding the Number of Maximal Groups}
\label{sub:Bounding_the_number_of_maximal_Groups}

\begin{wrapfigure}[12]{r}{0.4\textwidth}
  \centering
  \includegraphics{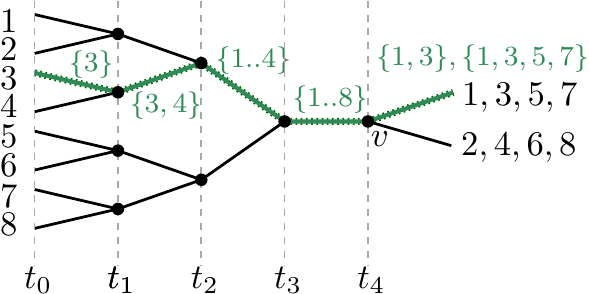}
  \caption{The maximal groups containing entity $3$ (green). Vertex $v$
    creates six new groups, including $\{1,3\}$ and $\{1,3,5,7\}$.}
  \label{fig:tree_example}
\end{wrapfigure}
To bound the total number of maximal groups, we study the case
where $m=1$ and $\delta=0$, because larger values can only reduce
the number of maximal groups.
It may seem as if each vertex in the Reeb graph simply creates as many
maximal groups as it has outgoing edges. However, consider for example
Fig.~\ref{fig:tree_example}. Split vertex $v$ creates not only the maximal
groups $\{1,3,5,7\}$ and $\{2,4,6,8\}$, but also $\{1,3\}$, $\{5,7\}$,
$\{2,4\}$, and $\{6,8\}$.
These last four groups are all maximal on $[t_2,t]$, for $t > t_4$.
Notice that all six newly discovered groups start strictly before $t_v$,
but only at $t_v$ do we realize that these groups are maximal, which
is the meaning that should be understood with ``creating maximal groups''.
This example can be extended to arbitrary size. Hence a vertex
 $v$ may create many new maximal groups, some of which start before~$t_v$.
We continue to show that we may obtain $\Omega(\tau n^3)$ maximal groups, and
that it cannot get worse than that, that is, the number of maximal groups is at
most $O(\tau n^3)$ as well.

\begin{lemma}
  \label{lem:lowerbound_num_groups}
  For a set \X of $n$ entities, in which each entity travels along a trajectory
  of $\tau$ edges, there can be $\Omega(n^3\tau)$ maximal groups.
\end{lemma}

\begin{proof}


  Similar to Lemma~\ref{lem:lowerbound_size_Reeb_graph} we construct $n$
  trajectory edges on which the entities travel in between $t_i$ and $t_{i+1}$,
  and repeat this construction in $O(\tau)$ time steps. Our construction yields
  $\Omega(n^3)$ maximal groups $G$ with $I_G \subseteq [t_i,t_{i+1}]$,
  resulting in $\Omega(\tau n^3)$ maximal groups overall as claimed.

  For ease of notation we assume that $n$ is divisible by four,
  and we write $x$ to denote both the entity $x$ and the $\eps$-disc of entity $x$.
  We partition our set of
  entities \X into two sets $S$ and $D$. The entities in $S = \{
  s_1,..,s_{3n/4} \}$ are stationary. They all lie on the line $y=0$, ordered
  from left to right, with a distance $r < 2\eps$ in between two
  consecutive entities. Hence $S$ is $\eps$-connected.

  The remaining
  entities $D$ will move on a horizontal line $y = \nu$, for some $\eps <
  \nu < 2\eps$. At time $t_i$, the discs $D = \{d_1,..,d_{n/4}\}$, ordered
  from right to left, all lie to the left of the discs in $S$. They all move to
  the right with the same speed. The distance $h_i$ between $d_i$ and
  $d_{i+1}$ is $r+(n/4 - i)\mu$, for some small $\mu > 0$. Hence, the distances
  get smaller the further the discs are to the left. See
  Fig.~\ref{fig:lowerbound_num_groups} for an illustration of this
  construction.

  \begin{figure}[hbt]
    \centering
    \includegraphics{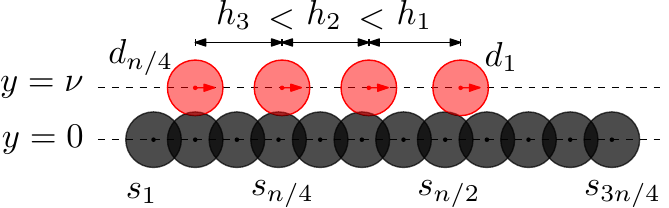}
    \caption{The lower bound construction for $n=16$. The black discs correspond
      to the stationary entities in $S$. The red (grey) discs correspond the
      entities in $D$.}
    \label{fig:lowerbound_num_groups}
  \end{figure}

  We can choose the exact values for $r$ and $\nu$ such that the sequence of
  events can be partitioned into \emph{rounds}. Round $i$ consists of
  a series of $k_i$ merge events followed by a series of $k_i$ split events.
  In a series $J_1,..,J_k$ of merges the discs $d_1,..,d_k$ become directly
  connected with discs in $S$. Merge $J_i$ will start a new maximal group
  $G_{1i}$, where $G_{ij} = S \cup \bigcup_{\ell=i}^j d_\ell$. Hence after the
  $k$ merges, $k$ maximal groups have started. In the subsequent series
  $P_1,..,P_k$ of split events, the discs $d_1,..,d_k$ stop being directly
  connected with a disc in $S$. When $d_i$ leaves, the sets of entities
  $G_{ii},..,G_{ik}$ end as maximal groups. However, when $d_i$ leaves
  $G_{ij}$, it creates $G_{(i+1)j}$ as a new maximal group that started on
  $J_{(i+1)}$ (see Fig.~\ref{fig:lowerbound_times}). This means $P_i$ creates
  $k - i$ new maximal groups.


  \begin{figure}[h]
    \centering
    \includegraphics{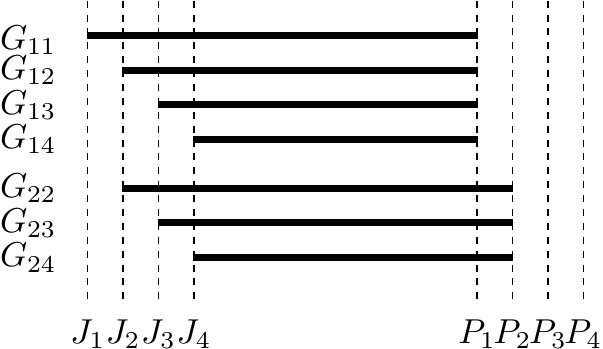}
    \caption{The time intervals on which $G_{ij}$ is a maximal group in a given
      round.}
    \label{fig:lowerbound_times}
  \end{figure}

  We now show that, for any $m \leq 3n/4$ and any $\delta$, this
  construction yields $\Omega(n^3)$ maximal groups. Since we can choose the
  speed of the discs in $D$, we can choose it such that all groups have
  a minimum duration of at least~$\delta$. Now consider the rounds
  $n/2,..,3n/4$. In each of these rounds we have $k = n/4$ merges followed
  by $n/4$ splits. The splits in each round create a total of
  $\sum_{i=1}^{n/4} (n/4-i) = \Omega(n^2)$ new maximal groups. Each of these
  groups contains $S$, hence its size is at least $3n/4$. It follows that the
  total number of maximal groups in those $n/4$ rounds is $\Omega(n^3)$.
\end{proof}

\begin{theorem}
  \label{th:upperbound_num_groups}
  Let \X be a set of $n$ entities, in which each entity travels along a
  trajectory of $\tau$ edges. There are at most $O(\tau n^3)$ maximal
  groups, and this is tight in the worst case.
\end{theorem}

\begin{proof}
  Lemma~\ref{lem:lowerbound_num_groups} 
  gives
  a construction that shows that there may be $\Omega(\tau n^3)$
  maximal groups. 

  We proceed with the upper bound.
  Every maximal group starts either at a start vertex, or a merge vertex.
  We will show that the number of maximal groups starting at a start or merge
  vertex is $O(n)$. Since there are $O(\tau n^2)$ start and merge vertices the lemma follows.
  We will discuss only the merge vertex case; the proof for a start vertex is the same.

  Let $v$ be a merge vertex, let $S\subset \X$ and $T\subset \X$ be the components
that merge at
  $v$, and let $p_x$ denote the path of entity $x \in S \cup T$ through \RR,
  starting at $v$. The union over all $x$ of these paths $p_x$ forms a
  directed acyclic graph (DAG)
  $\RR'_v$, which is a subgraph of \RR (see Fig.~\ref{fig:upperbound_num_groups}~(a)).
Consider ``unraveling'' $\RR'_v$ into a tree $\T_v$ as follows. If $p_x$ and
  $p_y$ split in some vertex $u$ and merge again in vertex $w$, with $t_w >
  t_u$ we duplicate the subpath starting at $w$. This yields a tree $\T_v$ with
  root $v$ and at most $|S|+|T| \leq n$ leaves. Furthermore, all nodes in $\T_v$
  have degree at most three (see Fig.~\ref{fig:upperbound_num_groups}~(b)).

  \begin{figure}[h]
    \centering
    \includegraphics{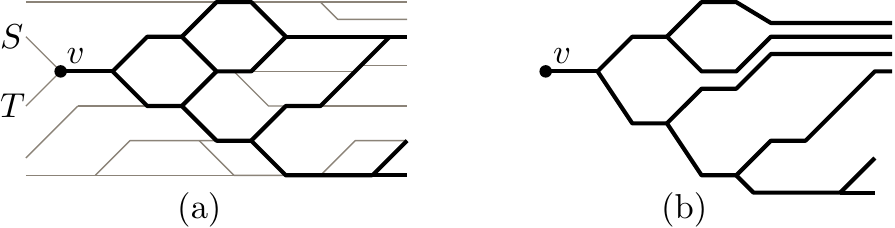}
    \caption{DAG $\RR'_v$ (black) as a subgraph of \RR (grey) (a), and the tree
      $\T_v$ obtained by unfolding $\RR'_v$ (b).}
    \label{fig:upperbound_num_groups}
  \end{figure}

  Since all maximal groups end at either a split or an end vertex, all
  maximal groups $G_1,..,G_k$ that start at $v$ can now be represented by
  subpaths in $\T_v$ starting at the root. The path corresponding to a maximal group $G$
  ends at the first node where two entities $x,y \in G$ split, or at a leaf if
  no such node exists. Clearly, paths $p_x$ and $p_y$ can split only at a
  degree three node. Since $\T_v$ has at most $n$ leaves it follows there are at
  most $O(n)$ degree three nodes.

  Finally, we show that there is at most one maximal group that ends at a given
  leaf or degree three node of $\T_v$. Assume by contradiction that $G_i$ and
  $G_j$, with $i \neq j$, both end at node $u$. Both maximal groups share the same path
  from the root of $\T_v$ to $u$, so all entities in $G_i$ and $G_j$ are in the
  same component at all times $t \in I = [t_v,t_u]$. Hence $G_i \cup G_j$ is a
  maximal group on $I$, contradicting that $G_i$ and $G_j$ were maximal. We
  conclude that the number of maximal groups $k$ that start at $v$ is at most
  the number of leaves plus the number of degree three nodes in $\T_v$. Hence $k =
  O(n)$. Summing over all $O(\tau n^2)$ start and merge vertices gives
  $O(\tau n^3)$ maximal groups in total. 
\end{proof}

\section{Computing the Grouping Structure}
\label{sec:Computing_the_Grouping_Structure}

To compute the grouping structure we need to compute the reduced Reeb graph
and the maximal groups.
We now show how to do this 
efficiently. 
Removing the edges of the Reeb graph that are not used is an easy post-processing step
which we do not discuss further.

\subsection{Computing the Reeb Graph}
\label{sub:Computing_the_Reeb-graph}


We can compute the Reeb graph $\RR = (V,E)$ as follows. We first compute all
times where two entities $x$ and $y$ are at distance $2\eps$ from each
other. We distinguish two types of events, \emph{connect events} at which $x$
and $y$ become directly connected, and \emph{disconnect events} at which $x$
and $y$ stop being directly connected.

We now process the events on increasing time while maintaining the current
components. We do this by maintaining a graph $G = (\X,Z)$ representing the
directly-connected relation, and the connected components in this graph. The
set of vertices in $G$ is the set of entities. The graph $G$ changes over time:
at connect events we insert new edges into $G$, and at disconnect events we
remove edges.

At any given time $t$, $G$ contains an edge $(x,y)$ if and only if $x$ and $y$
are directly connected at time~$t$. Hence the components at $t$ (the maximal
sets of $\eps$-connected entities) correspond to the connected components in
$G$ at time $t$. Since we know all times at which $G$ changes in advance,
we can use the same approach as \citet{parsa2012connectivity} to maintain the
connected components: we assign a weight to each edge in $G$
and we represent the connected components using a maximum weight spanning
forest. The weight of edge $(x,y)$ is equal to the time at which we remove it
from $G$, that is, the time at which $x$ and $y$ become directly
disconnected. We store the maximum weight spanning forest $F$ as an ST-tree
\cite{sleator1983sttrees}, which allows connectivity queries, inserts, and
deletes, in $O(\log n)$ time.

We spend $O(n^2)$ time to initialize the graph $G$ at $t_0$
in a brute-force manner.
For each component we create a start vertex in \RR.
We also initialize a one-to-one mapping $M$ from the
current components in $G$
to the corresponding vertices in \RR.
When we handle a connect event of entities $x$ and $y$ at time~$t$, we query
$F$ to get the components $C_x$ and $C_y$ containing $x$ and $y$,
respectively. Using $M$ we locate the corresponding vertices $v_x$ and $v_y$ in
\RR. If $C_x \neq C_y$ we create a new merge vertex $v$ in \RR with time $t_v =
t$, add edges $(v_x,v)$ and $(v_y,v)$ to \RR labeled $C_x$ and $C_y$,
respectively. If $C_x = C_y$ we do not change \RR. Finally, we add the edge
$(x,y)$ to $G$ (which may cause an update to $F$), and update the mapping $M$.

At a disconnect event we first query $F$ to find the component $C$ currently
containing $x$ and $y$. Using $M$ we locate the vertex $u$ corresponding to
$C$. Next, we delete the edge $(x,y)$ from $G$, and again query $F$. Let $C_x$
and $C_y$ denote the components containing $x$ and $y$, respectively. If $C_x =
C_y$ we are done, meaning $x$ and $y$ are still $\eps$-connected. Otherwise
we add a new split vertex $v$ to \RR with time $t_v = t$, and an edge $e =
(u,v)$ with $C_e = C$ as its component. We update $M$ accordingly.

Finally, we add an end vertex $v$ for each component $C$ in $F$ with $t_v=t_\tau$.
We connect the vertex $u = M(C)$ to $v$ by an edge $e = (u,v)$
and let $C_e=C$ be its component.

\paragraph{Analysis} We need $O(\tau n^2\log n)$ time to compute all $O(\tau
n^2)$ events and sort them according to increasing time. To handle an event we
query $F$ a constant number of times, and we insert or delete an edge in
$F$. These operations all take $O(\log n)$ time. So the total time required
for building \RR is $O(\tau n^2 \log n)$.
\begin{theorem}
  \label{thm:compute_Reeb_graph}
  Given a set \X of $n$ entities, in which each entity travels along a
  trajectory of $\tau$ edges, the Reeb graph $\RR = (V,E)$ has $O(\tau n^2)$
  vertices and edges, and can be computed in $O(\tau n^2 \log n)$
  time.
\end{theorem}

\subsection{Computing the Maximal Groups}

We now show how to compute all maximal groups using the Reeb graph $\RR =
(V,E)$. We will ignore the requirements that each maximal group
should contain at least $m$ entities and have a minimal duration of
$\delta$. That is, we assume $m=1$ and $\delta=0$.
It is easy to adapt the algorithm for larger values.

\paragraph{Labeling the Edges} Our algorithm labels each edge $e = (u,v)$ in
the Reeb graph with a set of maximal groups \groups{e}. The groups $G \in
\groups{e}$ are those groups for which we have discovered that $G$ is a maximal
group at a time $t \leq t_u$. Each maximal group $G$ becomes maximal at a
vertex, either because a merge vertex created $G$ as a new group that is maximal,
or because $G$ is now a maximal set of entities that is still together after a
split vertex. This means we can compute all maximal groups as follows.

We traverse the set of vertices of \RR in topological order. For every vertex $v$ we
compute the maximal groups on its outgoing edge(s) using the information on its
incoming edge(s).

If $v$ is a start vertex it has one outgoing edge $e = (v,u)$. We set
\groups{e} to $\{(\comp{e},t_v)\}$ where $t_v=t_0$.
If $v$ is a merge vertex it has two incoming edges, $e_1$ and $e_2$. We
propagate the maximal groups from $e_1$ and $e_2$ on to the outgoing edge $e$, and we
discover $(\comp{e},t_v)$ as a new maximal group. Hence $\groups{e} = \groups{e_1} \cup
\groups{e_2} \cup \{(\comp{e},t_v)\}$.

\begin{wrapfigure}[12]{r}{0.46\textwidth}
  \centering
  \includegraphics{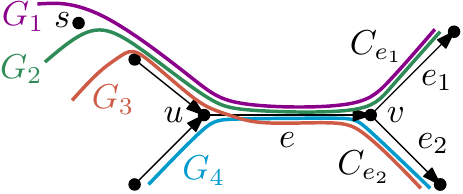}
  \caption{After split vertex $v$, \groups{e_1} contains the groups $C_{e_1} =
    G_1 \cup G_2$ (with starting time $t_s$), $G_1$, and $G_2$. Maximal groups $C_{e_2}
    = G_3 \cup G_4$ (with starting time $t_u$), $G_3$, and $G_4$ go to
    $e_2$. The maximal groups $C_e$ and $G_1 \cup G_2 \cup G_3$ end at $v$. }
  \label{fig:split_vertex}
\end{wrapfigure}
If $v$ is a split vertex it has one incoming edge $e$, and two outgoing edges
$e_1$ and $e_2$. A maximal group $G$ on 
$e$ may end at $v$, continue
on $e_1$ or $e_2$, or spawn a new maximal group $G' \subset G$ on either
$e_1$ or $e_2$. In particular, for any group $G'$ in \groups{e_i}, there is a
group $G$ in \groups{e} such that $G' = G \cap C_i \neq \emptyset$. The
starting time of $G'$ is $t' = \min \{ t \mid (G,t) \in \groups{e} \land G'
\subseteq G \}$. Thus, $t'$ is the first time $G'$ was part of a maximal group on
$e$. Stated differently, $t'$ is the first time $G'$ was in a component on a path
to $v$. Fig.~\ref{fig:split_vertex} illustrates this case.
If $v$ is an end vertex it has no outgoing edges. So there is nothing to be
done.

Fig.~\ref{fig:example} shows a complete example of a Reeb graph after labeling
the edges with their maximal groups.

\begin{figure*}[b]
  \centering
  \includegraphics{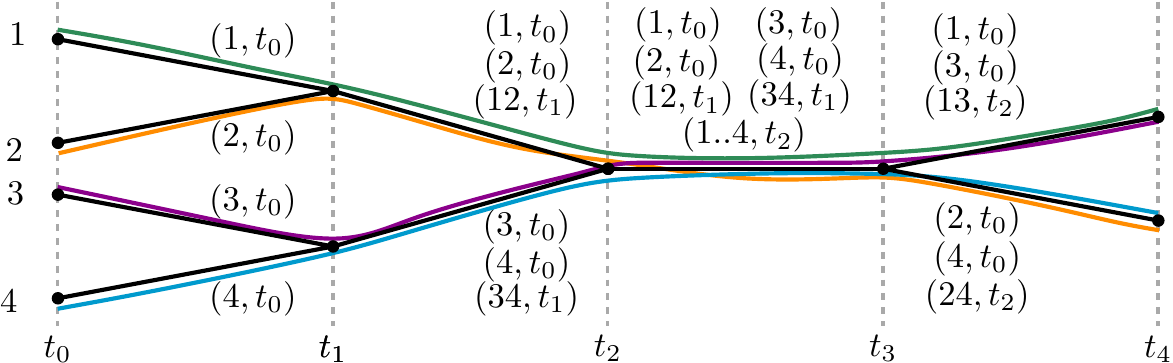}
  \caption{The maximal groups as computed by our algorithm (a set $\{i,j,k\}$
    is denoted by $ijk$). }
  \label{fig:example}
\end{figure*}

\paragraph{Storing the Maximal Groups} We need a way to store the maximal
groups \groups{e} on
an edge $e=(u,v)$ in such a way that we can efficiently compute the set(s) of
maximal groups on the outgoing edge(s) of a vertex $v$. We now show that we can use a
tree $\T_e$ to represent \groups{e}, with which we can handle a merge vertex in $O(1)$
time, and a split vertex in $O(k)$ time, where $k$ is the number of entities
involved. The tree uses $O(k)$ storage.

We say a group $G$ is a \emph{subgroup} of a group $H$ if and only if
$G\subseteq H$ and $I_H\subseteq I_G$. For example, in
Fig.~\ref{fig:example_groups} $\{x_1,x_2\}$ is a subgroup of $\{x_1,..,x_4\}$.
Note that both $G$ and $H$ could be maximal.

\begin{lemma}
  \label{lem:candidate_groups_are_subsets}
  Let $e$ be an edge of \RR, and let $S$ and $T$ be maximal groups in
  \groups{e} with starting times $t_S$ and $t_T$, respectively.  There is also
  a maximal group $G \supseteq S\cup T$ on $e$ with starting time $t_G \geq
  \max(t_S,t_T)$, and if $S \cap T \neq \emptyset$ then $S$ is a
  subgroup of $T$ or vice versa.
\end{lemma}

\begin{proof}
  The first statement is almost trivial. Clearly, $S,T \subseteq C_e$ and hence
  $S \cup T \subseteq C_e$. Component $C_e$ itself is also a maximal group on $e$.
  By construction $C_e$ must have the largest starting time $t$ of the groups in
  \groups{e}. Hence $t_G \geq \max(t_S,t_T)$.

  We prove the second statement by contradiction: assume $S \cap T \neq
  \emptyset$, and $S \not \subseteq T$ or vice versa. Assume w.l.o.g.\ that $t_S
  \leq t_T$. So the entities in $S$ are all in a single component at all times
  $t \geq t_T\geq t_S$. At any time $t \geq t_T$ all entities in $T$ are also in a
  single component. Since $S \cap T \neq \emptyset$ this must be the same
  component that contains $S$. Hence $S \subseteq T$, which together
  with $t_S \leq t_T$ proves the statement.
\end{proof}
We represent the groups \groups{e} on an edge $e \in E$ by a tree $\T_e$ (see
Fig.~\ref{fig:grouping_tree}). We call
this the \emph{grouping tree}. Each node $v$ represents a group $G_v \in
\groups{e}$. The children of a node $v$ are the largest subgroups of
$G_v$. From Lemma~\ref{lem:candidate_groups_are_subsets} it follows that any
two children of $v$ are disjoint. Hence an entity $x \in G_v$ occurs in only
one child of $v$. Furthermore, note that the starting times are monotonically
decreasing on the path from the root to a leaf: smaller groups started
earlier. A leaf corresponds to a smallest maximal group on $e$: a singleton
set with an entity $x \in C_e$. It follows that $\T_e$ has $O(n)$ leaves,
and therefore has size $O(n)$. Note, however, that
the summed sizes of all maximal groups can be quadratic.

\begin{wrapfigure}[9]{r}{0.35\textwidth}
  \centering
  \vspace{-.5\baselineskip}
  \includegraphics{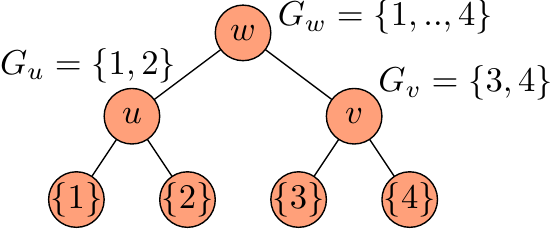}
  \caption{The grouping tree for the edge between $t_2$ and $t_3$ in
    Fig.~\ref{fig:example}.}
  \label{fig:grouping_tree}
\end{wrapfigure}

\paragraph{Analysis} We analyze the time required to label each edge $e$ with a
tree $\T_e$ for a given Reeb graph $\RR = (V,E)$. Topologically sorting the
vertices takes linear time. So the running time is determined by the processing
time in each vertex, that is, computing the tree(s) $\T_e$ on the outgoing
edge(s) $e$ of each vertex. Start, end, and merge vertices can be handled in $O(1)$
time: start and end vertices are trivial, and at a merge vertex $v$ the tree
$\T_e$ is simply a new root node with time $t_v$ and as children the (roots of the)
trees of the incoming edges. At a split vertex we have to split the tree $\T = \T_{(u,v)}$
of the incoming edge $(u,v)$ into two trees for the outgoing edges of $v$. For this,
we traverse $\T$ in a bottom-up fashion, and for each node, check whether it induces
a vertex in one or both of the trees after splitting.
This algorithm runs in $O(|\T|)$ time. Since $|\T| = O(n)$ the total running time of our
algorithm is $O(n|V|) = O(\tau n^3)$.

\paragraph{Reporting the Groups}
We can augment our algorithm to report all maximal groups at split and end vertices.
The main observation is that a maximal group ending at a split vertex $v$, corresponds
exactly to a node in the tree $\T_{(u,v)}$ (before the split) that has entities
in leaves below it that separate at $v$.
The procedures for handling split and end vertices can easily
be extended to report the maximal groups of size at least $m$ and duration at
least $\delta$ by simply checking this for each maximal group.
Although the number of maximal groups is $O(\tau n^3)$
(Theorem~\ref{th:upperbound_num_groups}), the summed size of all maximal groups
can be $\Omega(\tau n^4)$.
The running time of our algorithm is $O(\tau n^3+N)$, where $N$ is the total output size.

\begin{theorem}
  \label{thm:compute_groups}
  Given a set \X of $n$ entities, in which each entity travels along a
  trajectory of $\tau$ edges, we can compute all maximal groups in
  $O(\tau n^3+N)$ time, where $N$ is the output size.
\end{theorem}

\section{Robustness}
\label{sec:Robustness}

The grouping structure definition we have given and analyzed has a number
of good properties. It fulfills monotonicity,
and in the previous sections we showed that there are only polynomially
many maximal groups, which can be computed in polynomial time as well.
In this section we study the property of robustness, which our definition
of grouping structure does not have yet.
Intuitively, a robust grouping structure ignores short interruptions of
groups, as these interruptions may be insignificant at the
temporal scale at which we are studying the data. For example, if we
are interested in groups that have a duration of one hour or more, we may
want to consider interruptions of a minute or less insignificant.

We introduce a new temporal parameter $\alpha$, which is related to the
temporal scale at which the data is studied. Our robust grouping structure
should ignore interruptions of duration at most $\alpha$. We realize this by
letting the precise moment of events be irrelevant beyond a value depending on
$\alpha$. Events that happen within $\alpha$ time of each other may cancel out,
or their order may be exchanged.  The objective is to incorporate $\alpha$ into
our definitions while maintaining the properties that we have for the
(non-robust) grouping structure.  Note that $\alpha$ is another parameter that
allows us to obtain more generalized views of the grouping structure by
increasing its value.  Obtaining generalized views in this way is related to
the concept of persistence in computational
topology~\cite{eh-cti-10,edelsbrunner2002persistence}.


A possible definition of a robust grouping structure is based on the following
intuition: A set of entities forms a robust group on $I$ as long as every
interval $I' \subset I$ on which its entities are not in the same component
has length at most $\alpha$. More formally: we say $G$ is a \emph{robust group}
on time interval $I$ if and only if: (\textit{i}) $G$ contains at least $m$
entities, (\textit{ii}) $I$ has length at least $\delta$, and (\textit{iii})
for any time $t \in I$ there is a time $t' \in [t-\alpha/2,\; t+\alpha/2]$ and
a component $C \in \C(t')$ such that $G \subseteq C$.
Unfortunately, we can show that even determining whether there is a robust group of
size $k$ is NP-complete (see Appendix~\ref{app:np_complete}).

We consider a second definition for a robust group, which we will use
from now on. Two entities are \emph{$\alpha$-relaxed directly
  connected} at time $t$ if and only if they are directly connected at some
time $t'\in [t-\alpha/2,\; t+\alpha/2]$. Two entities $x$ and $y$ are
\emph{$\alpha$-relaxed $\eps$-connected} at time $t$ if there is a sequence
$x=x_0,..,x_j=y$ such that $x_i$ and $x_{i+1}$ are $\alpha$-relaxed directly
connected. Note that the precise times may be different for different pairs
$x_i$ and $x_{i+1}$, as long as each time is in the interval $[t-\alpha/2,\;
t+\alpha/2]$. A maximal set of $\alpha$-relaxed $\eps$-connected entities at
time $t$ is an \emph{$\alpha$-relaxed component}, or \emph{$\alpha$-component}
for short.  An $\alpha$-component at time $t$ corresponds to connected
$3D$-component in a horizontal slice of \M with thickness $\alpha$ and centered at $t$
(see Fig.~\ref{fig:alpha_component}).


\begin{wrapfigure}[7]{r}{0.4\textwidth}
  \centering
  \vspace{-1.25\baselineskip}
  \includegraphics{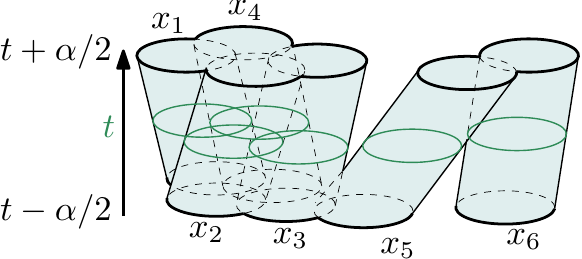}
  \vspace{-.75\baselineskip}
  \caption{An $\alpha$-component at time $t$.} 
  \label{fig:alpha_component}
\end{wrapfigure}
A subset $G$ of $k$ entities is a \emph{robust group} if and only if it is a
group by the definition in the introduction, but where ``component'' is
replaced by ``$\alpha$-component'' in condition (\textit{iii}). This
immediately leads to the definition of maximal robust groups and a robust
grouping structure. The robust grouping structure has the property of
monotonicity in the new parameter $\alpha$ as well.
Note that every group which is a robust group according to the first definition, is also a robust group according to the second definition.
%
For instance, in Fig.~\ref{fig:alpha_component},
entities $x_1,..,x_6$ form a component by the second definition, but not by the first.


\subsection{Computation of Maximal Robust Groups}

We can compute all maximal robust groups according to the
(second) definition. The idea is to modify the Reeb graph
to a version that is parametrized by $\alpha$ and captures exactly the
robust grouping structure for parameter $\alpha$.

Let $\RR$ be the Reeb graph that we used for the grouping structure
without considering robustness. Note that this is the same as assuming
$\alpha=0$ in the definition of the robust grouping structure, and we
let $\RR_0=\RR$.
For $\alpha>0$ we define the Reeb graph parametrized in $\gamma$ as
$\RR_\gamma$ 
by imagining a process that changes the Reeb graph for a growing
parameter $\gamma$, starting with $\RR_0$ and ending with $\RR_{\alpha/2}$.

We observe that a new $\alpha$-component starts at time $\alpha/2$ before two
regular components merge and form a new component.
Symmetrically, an $\alpha$-component ends due to a split at time $\alpha/2$
after a regular component splits. Both facts follow from the new definition
of $\alpha$-relaxed directly connected. It implies that in the process
that maintains $\RR_\gamma$ for growing $\gamma$,
the split nodes move forward in time, zippering together the outgoing edges,
and the merge nodes move backward in time, zippering together the incoming edges.
All nodes move at the same rate in $\gamma$, which implies that in the process,
the only event where the Reeb graph changes structurally is when an (earlier)
split node encounters a (later) merge node.
This can happen only if they are endpoints of the same edge of the Reeb graph.
The encounter is either a \emph{passing} or a \emph{collapse} (see
Fig.~\ref{fig:encounter}).

\begin{wrapfigure}[9]{r}{0.3\textwidth}
  \centering
  \vspace{-.75\baselineskip}
  \includegraphics{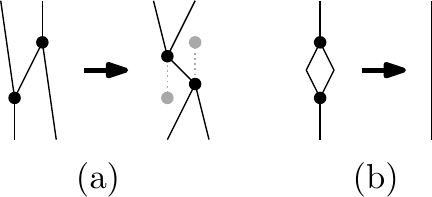}
  \caption{Passing encounter, before and after (a). Collapse encounter,
before and after (b).}
  \label{fig:encounter}
\end{wrapfigure}
Both encounters lead to new edges in the Reeb graph and can thus give rise
to new encounters
when growing $\gamma$ further.  The collapse
encounter reduces the complexity of the Reeb graph: two nodes of degree $3$
disappear and four edges become a single edge. The collapse event is exactly
the situation where a component splits and merges again, so by removing a
split-merge pair involving the same entities we ignore the
temporary split of a component (or group). 

A passing encounter maintains the complexity of the Reeb graph. Before the passing
encounter, a part of one group splits and merges with a different group.
After the passing encounter, the two groups merge (for a short time) and then split again.
This situation is also captured in Fig.~\ref{fig:alpha_component}.

Next, we show that there are $O(\tau n^3)$ encounter events in the Reeb graph
of the robust version of the trajectory grouping structure, and this bound is
tight in the worst case.

\begin{lemma}
  \label{lem:lowerbound_encounters}
  For some set \X of $n$ entities, in which each entity travels along a trajectory
  of $\tau$ edges, the structure of the Reeb graph $\RR_\gamma$ of \X
  changes $\Omega(\tau n^3)$ times when increasing $\gamma$ from zero to
  infinity.
\end{lemma}

\begin{proof}
  We show that there is a set of $n$
  trajectories, each consisting of $\tau$ edges, for which there are
  $\Omega(\tau n^3)$ encounter events. The lemma then follows.

  We use the same construction as in Lemma~\ref{lem:lowerbound_num_groups}. So in all time
  intervals $[t_{2i},t_{2i+1}]$ we have a set $S$ of $3n/4$ stationary
  entities/discs and a set $D=\{d_1,..,d_{n/4}\}$ entities, ordered from right
  to left, that move to the right in such a way that $d_i$ becomes directly
  (dis)connected with $S$ before $d_{i+1}$ (see
  Fig.~\ref{fig:lowerbound_num_groups}). Let $t_a$ be the first time at which
  $d_{n/4}$ becomes directly connected with $S$, and let $t_b$ denote the last
  time $d_1$ becomes directly disconnected with $S$. We now show that the part
  of Reeb-graph $\RR'$ corresponding to the interval $(t_a,t_b)$ already yields
  $\Omega(n^3)$ encounter events. We note that no other encounter events
  involving other parts of the Reeb-graph can interfere with the encounter
  events in $\RR'$.

  In between $t_a$ and $t_b$ every disc $d_i$ becomes directly (dis)\-connected
  with $S$ $\Omega(n)$ times. So $\RR'_\gamma$ initially contains of a path $P$
  of $\Omega(n^2)$ edges. Each edge has at least the set of entities $S$
  associated with it, and possibly other entities as well. The vertices on $P$
  can be grouped in $\Omega(n)$ sequences of $k=n/4$ split vertices
  $u_1,..,u_k$ followed by $k$ merge vertices $v_1,..,v_k$. At vertex $u_i$
  entity $x_i$ splits from $S$ and at $v_i$ entity $x_i$ merges with $S$. See
  Fig.~\ref{fig:lowerbound_passing_encounters}.

  \begin{figure}[h]
    \centering
    \includegraphics{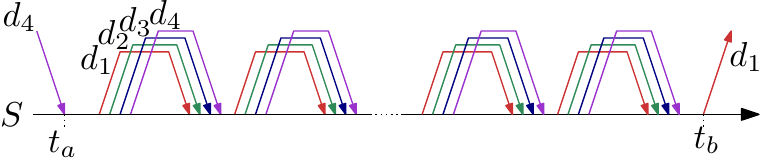}
    \caption{The part of the Reeb-graph that yields $\Omega(n^3)$ encounter
      events (for $n=16$). }
    \label{fig:lowerbound_passing_encounters}
  \end{figure}

  By increasing $\gamma$ each split vertex $u_i$ will have a passing encounter
  with the merge vertices $v_1,..,v_{i-1}$ before it collapses with
  $v_i$. Hence each sequence involves $\sum_{i=1}^k (i-1) = \Omega(n^2)$
  encounter events. Since there are $\Omega(n)$ such sequences this gives
  $\Omega(n^3)$ encounter events in a single timestep, and hence $\Omega(\tau
  n^3)$ in total.
\end{proof}

\begin{theorem}
  \label{thm:upperbound_encounters}
  Let \X be a set of $n$ entities, in which each entity travels along a trajectory of
  $\tau$ edges. The structure of the Reeb graph $\RR_\gamma$ of \X changes at
  most $O(\tau n^3)$ times when increasing $\gamma$ from zero to infinity.
  This bound is tight in the worst case.
\end{theorem}

\begin{proof}
  Lemma~\ref{lem:lowerbound_encounters} gives
  a construction that shows that there may be $\Omega(\tau n^3)$
  encounters.

  Since each collapse event decreases the number of edges by three it follows
  the number of collapse events is at most $O(\tau n^2)$. What remains is to
  prove that the number of passing events is $O(\tau n^3)$. Each passing event
  involves a split vertex $u$ and a merge vertex $v$. We now show that there
  are at most $n$ passing events involving a given split vertex $u$. Since
  there are $O(\tau n^2)$ split vertices this means the number of passing
  events is $O(\tau n^3)$.

  Assume by contradiction that there are $k > n$ passing events involving split
  vertex $u$. Let $\gamma_1,..,\gamma_k$ be the values for $\gamma$ for which
  these passing events occur in non-decreasing order, and let $v_1,..,v_k$ be
  the corresponding merge vertices. Just before $u$ passes $v_i$ the edge
  $e=(u,v_i)$ is an incoming edge of $v_i$. Let $X_i$ denote the set of
  entities on the other incoming edge of $v_i$, that is the set of entities
  that merges with $C_e$ at vertex $v_i$ (see
  Fig.~\ref{fig:upperbound_passing_encounters}(a)).

  Since $k > n$ there must be an entity $x$ that $u$ ``passes'' at least
  twice. That is, $u$ passes $v_i$ and $v_j$, with $i < j$, and $x \in X_i$ and
  $x \in X_j$. Now consider the Reeb-graph $\RR_\gamma$ just after $u$ passes
  $v_i$ (which means $\gamma > \gamma_i$). Since $u$ still has to pass $v_j$
  there is a path $Q$ connecting $u$ to $v_j$. By further increasing $\gamma$
  this path will eventually become a single edge $(u,v_j)$, which will flip to
  $(v_j,u)$ when $u$ passes $v_j$ at $\gamma = \gamma_j$.

  Entity $x$ is present at the first vertex of $Q$ (vertex $u$), and it merges
  again with path $Q$ at $v_j$. Clearly, this means that $Q$ contains a split
  vertex $w$ at which $x$ splits from path $Q$ before it can return to $Q$ in
  vertex $v_j$ (see Fig.~\ref{fig:upperbound_passing_encounters} (b)).

  \begin{figure}[h]
    \centering
    \includegraphics{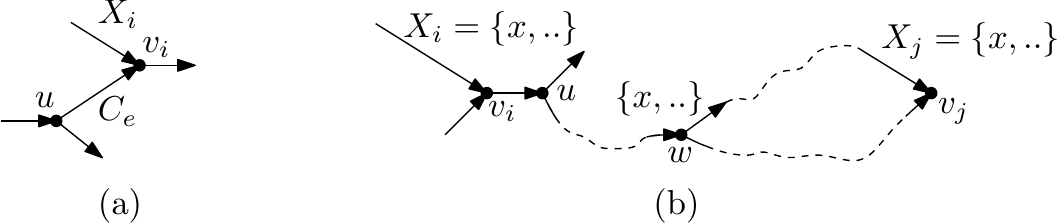}
    \caption{The part of $\RR_\gamma$ before $u$ encounters $v_i$. The set
      $X_i$ merges with $C_e$ at vertex $v_i$ (a). If $x$ merges at both $v_i$
      and $v_j$ it has to leave (split) at a vertex $w$ in between (b).}
    \label{fig:upperbound_passing_encounters}
  \end{figure}

  We now have two paths connecting $w$ to $v_j$: the path that $x$ follows and
  the subpath of $Q$. We again have that by increasing $\gamma$ both paths will
  become singleton edges connecting $w$ to $v_j$. Eventually both these edges
  are removed in a collapse event for some $\hat{\gamma}$. If $w = u$ this means
  $(u,v_j)$ is actually a collapse event instead of a passing
  event. Contradiction. If $w \neq u$ we have that $t_w > t_u$, and therefore
  $\hat{\gamma} < \gamma_j$. The collapse event at $\hat{\gamma}$ will consume both $w$
  and $v_j$, which means $u$ can no longer pass $v_j$. Contradiction. Since
  both cases yield a contradiction we conclude that the number of passing
  events involving $u$ is at most $n$. With $O(\tau n^2)$ vertices this yields
  the desired bound of $O(\tau n^3)$ passing events.
\end{proof}

Algorithmically, we start with the Reeb graph $\RR_0$ and examine each edge.
Any edge that leads from a split node to a merge node and whose duration
is at most $\alpha$ is inserted in a priority queue, where the duration of
the edge is the priority. We handle the encounter events in the correct order,
changing the Reeb graph and possibly inserting new encounter events in the
priority queue. Each event is handled in $O(\log n)$ time since it involves
at most $O(1)$ priority queue operations. Since there are $O(\tau n^3)$ events
(Theorem~\ref{thm:upperbound_encounters}) this takes $O(\tau n^3\log n)$ time in
total.
Once we have the Reeb graph $\RR_{\alpha/2}$, we can associate the trajectories
with its edges as before. The computation of the maximal robust groups is
done in the same way as computing the maximal groups on the
normal Reeb graph $\RR$. We conclude:

\begin{theorem}
  \label{thm:compute_robust_groups}
  Given a set \X of $n$ entities, in which each entity travels along a
  trajectory of $\tau$ edges, we can compute all robust maximal groups in
  $O(\tau n^3 \log n +N)$ time, where $N$ is the output size.
\end{theorem}

\section{Evaluation}
\label{sec:Evaluation}

To see if our model of the grouping structure is practical and indeed captures the grouping
behavior of entities we implemented and evaluated our algorithms. We would
like to visually inspect the maximal groups identified by our algorithm, and
compare this to our intuition of groups. For a small number of (short)
trajectories we can still show this in a figure, see for example
Fig.~\ref{fig:small_multiples}, which shows the monotonicity of the maximal
groups in size and duration. However, for a larger number of trajectories the
resulting figures become too cluttered to analyze. So instead we generated short
videos.\footnote{See \url{www.staff.science.uu.nl/~staal006/grouping}.}

\begin{figure*}[t]
  \centering
  \includegraphics[width=\textwidth]{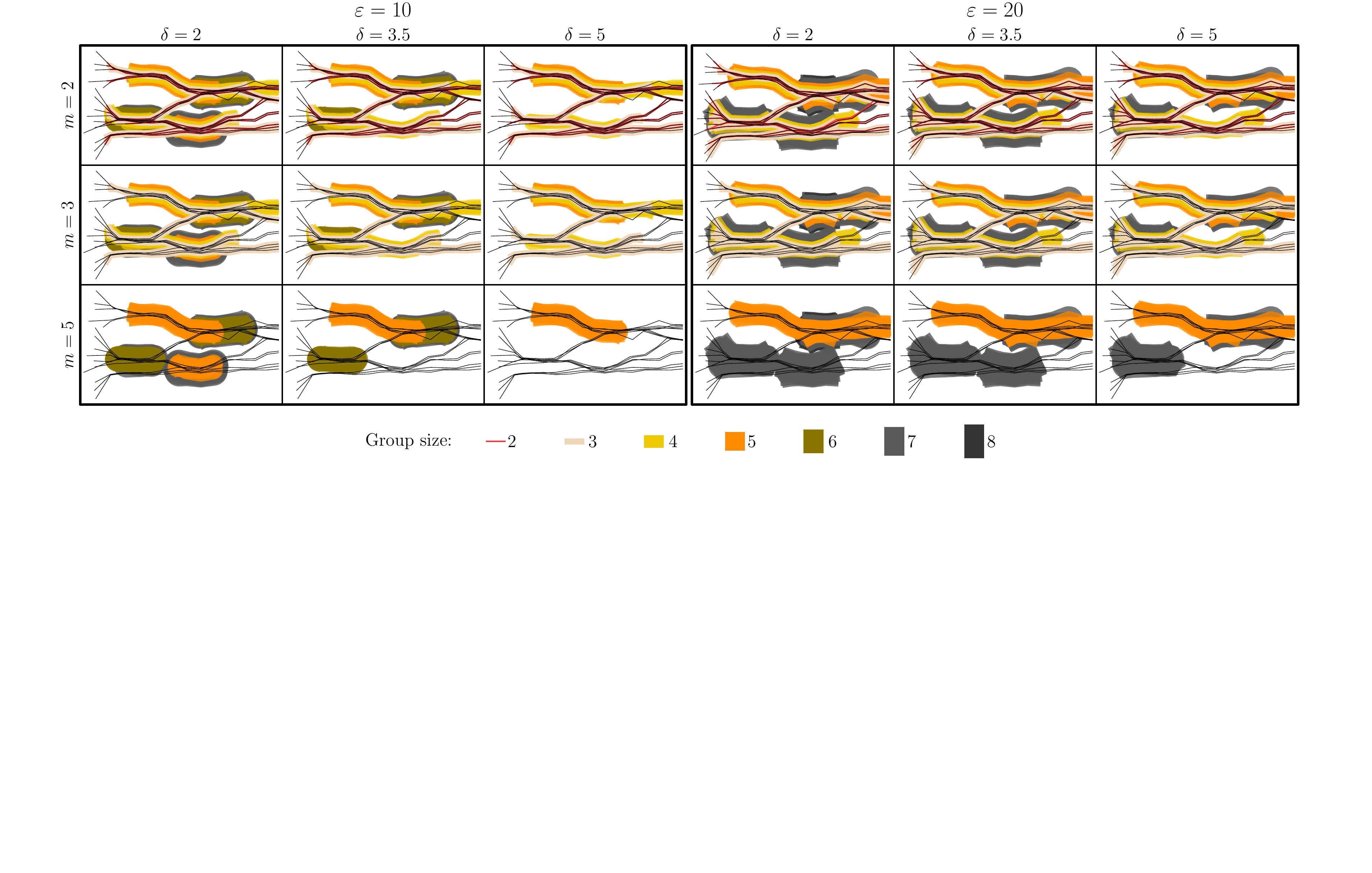}
  \caption{The maximal groups for varying parameter values. The time associated
with each trajectory vertex is proportional to its $x$-coordinate.}
  \label{fig:small_multiples}
\end{figure*}

We use two types of data sets to evaluate our method: a synthetic data set
generated using a slightly modified version of the NetLogo Flocking model
\cite{netlogo_flocking,netlogo}, and a real-world data set consisting of
deer, elk, and cattle, tracked in the Starkey project \cite{starkey}.



\paragraph{NetLogo} We generated several data sets using an adapted version of
the NetLogo Flocking model~\cite{netlogo_flocking}. In our adapted model the
entities no longer wrap around the world border, but instead start to turn when
they approach the border. Furthermore, we allow small random direction changes
for the entities. The data set that we consider here contains 400 trajectories,
with 818 edges each. Similar to Fig.~\ref{fig:small_multiples}, our videos
show all maximal groups for varying parameter values.

The videos show that our model indeed captures the crucial properties of grouping behavior well. We notice that the choice of  parameter values is important. In
particular, if we make $\eps$ too large we see that the entities are loosely
coupled, and too many groups are found. Similarly, for large values of $m$
virtually no groups are found. However, for reasonable parameter settings, for
example $\eps = 5.25$, $m=4$, and $\delta = 100$, we can clearly see that our
algorithm identified virtually all sets of entities that travel
together. Furthermore, if we see a set of entities traveling together that is
not identified as group, we indeed see that they disperse quickly after they
have come together. The coloring of the line-segments also nicely shows how
smaller groups merge into larger ones, and how the larger groups break up into
smaller subgroups. This is further evidence that our model captures the grouping behavior well.

\paragraph{Starkey} We also ran our algorithms on a real-world data
set, namely on tracking data obtained in the Starkey project~\cite{starkey}. This
data set captures the movement of deer, elk, and cattle in Starkey, a large
forest area in Oregon (US), over 
three years. Not all animals are
tracked during the entire period, and positions are not reported
synchronously for all entities. Thus, we consider only a subset of the
data, and resample the data such that all trajectories have vertices at the
same (regularly spaced) times. We chose a period of 30 days for which we have
the locations of most of the animals. This yields a data set containing 126
trajectories with 1264 vertices each.
In the Starkey video we can see that a large group of entities quickly forms in
the center, and then slowly splits into multiple smaller groups. We notice that
some entities (groups) move closely together, whereas others often stay
stationary, or travel separately.

\paragraph{Running Times} Since we are mainly interested in how well our
model captures the grouping behavior, we do not extensively evaluate the
running times of our algorithms. On our desktop system with a AMD Phenom II X2
CPU running at 3.2Ghz our algorithm, implemented in Haskell, computes the
grouping structure for our data sets in a few seconds. Even for 160
trajectories with roughly 20 thousand vertices each we can compute and report
all maximal groups in three minutes. Most of the time is spent on computing the
Reeb graph, in particular on computing the connect/disconnect events. Since our
implementation uses a slightly easier, yet slower, data structure to represent
the maximum weight spanning forest during the construction of the Reeb graph,
we expect that some speedup is still possible.

\section{Concluding Remarks}
\label{sec:Concluding_Remarks}

We introduced a trajectory grouping structure which uses Reeb graphs and a notion of persistence for robustness. We showed how to characterize and efficiently compute the maximal groups and group changes in a set of trajectories, and bounded their maximal number.
Our paper demonstrates that computational topology provides a mathematically sound way to
define grouping of moving entities. The complexity bounds, algorithms and implementation
together form the first comprehensive study of grouping.
Our videos show that our methods produce results that correspond to human intuition.

Further work includes more extensive experiments together with domain specialists,
such as behavioral biologists, to ensure further that the grouping structure
captures groups and events in a natural, expected way, and changes in the
parameters have the desired effect. At the same time, our research may be linked
to behavioral models of collective motion~\cite{s-cab-10}
and provide a (quantifiable) comparison of these.

We expect that for realistic inputs the size of the grouping structure
is much smaller than the worst-case bound that we proved. We plan to confirm
this in experiments, and to provide faster algorithms under realistic input models.
We will also work on improving the visualization of the maximal groups and the
grouping structure, based on the reduced Reeb graph.




\bibliographystyle{abbrvnat}
\bibliography{grouping}

\bigskip\bigskip\bigskip

\noindent
Videos accompanying this paper can be found on \url{www.staff.science.uu.nl/~staal006/grouping}.

\clearpage
\appendix





\section{NP-completeness of robust grouping by the first definition}
\label{app:np_complete}

\begin{theorem}
  \label{thm:robust_groups_np_hard}
  Determining whether there is a robust group of size $k$ is NP-complete using
  the first definition of robust groups.
\end{theorem}

\begin{proof}
  We prove this by a reduction from \problem{Clique}: given a graph $G=(V,E)$
  is there a clique of size $k$? Choose $\eps = 0$, $m \leq k$, $\delta
  \leq n+1$, and $\alpha = 3/4$. We now construct a set of $n$ trajectories,
  one for each vertex, each consisting of $O(n)$ vertices such that there is a
  robust group $R$ on $I=[1,n+1]$ consisting of $k$ entities if and only if $G$
  contains a clique $R'$ of size $k$. The proof idea is similar to that in
  \cite{gudmundsson2006computing}.

  Let $N(v)$ denote the neighbours of vertex $v \in V$. For each vertex $v_i$
  we define five points $p_i, a_i, b_i, c_i,$ and $d_i$. Additionally, we
  define a point $p_{n+1}$.  We assume that all these points (over all vertices)
  are different. Let $s_i = (i+1)-\alpha = i+(1/4)$ and $t_i=i+\alpha =
  i+(3/4)$ be two times corresponding to vertex $v_i$. We now construct an
  entity/trajectory $x_i$ for each vertex $v_i \in V$ such that:
  \begin{itemize*}
  \item at time $j$, $x_i$ is at $p_j$,
  \item at time $s_j$, $x_i$ is at $a_j$ if $v_i = v_j$, and at $b_j$
    otherwise,
  \item at time $t_j$, $x_i$ is at $c_j$ if $v_i \in \{v_j\} \cup N(v_j)$, and
    at $d_j$ otherwise, and
  \item at any other time no two entities are at the same place at the same
    time.
  \end{itemize*}
  Fig.~\ref{fig:np_hardness_example} shows an example of this construction.

  \begin{figure*}[ht]
    \centering
    \includegraphics{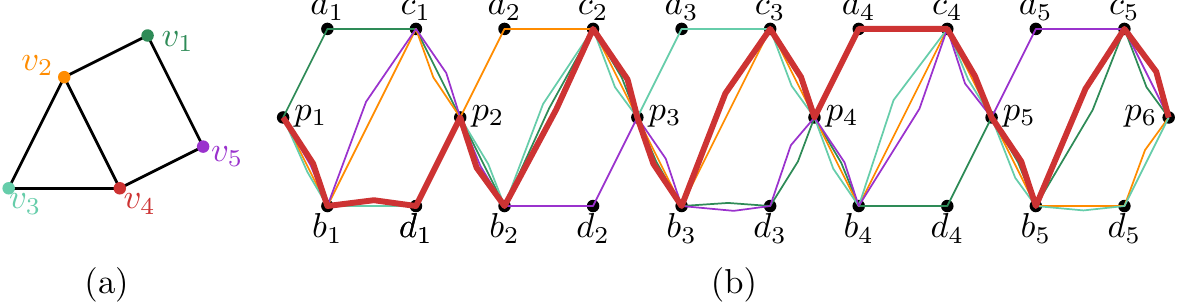}
    \caption{An input graph $G=(V,E)$ (a), the trajectories for $G$, the
      $x$-coordinate of the points corresponds to the time (b). The trajectory
      corresponding to $v_4$ is shown in bold.}
    \label{fig:np_hardness_example}
  \end{figure*}

  Since $\eps$ is set to zero all entities in a robust group $R$ have to
  be at the same point in every interval of length $\alpha$. The only times
  when multiple entities are at the same point are at times $i,s_i$ $t_i$, with
  $1 \leq i \leq n+1$. Because $i+1-i > \alpha$ it follows all entities in $R$
  have to be together at $s_i$ or $t_i$. We now select a vertex to be part of
  the clique $R'$ if and only if the entities in $R$ were not together at time
  $s_i$. All entities except $x_i$ are together at time $s_i$, so it follows
  that $x_i \in R$. We then have $R' = \{v_i \mid x_i \in R\}$.

  Suppose there is a robust group $R$ of size $k$ on $I$. We now show that for
  every pair $v_i,v_j \in R'$, $v_i$ and $v_j$ are neighbours. Hence $R'$ forms
  a clique (of size $k$).

  Both $v_i$ and $v_j$ are in $R'$, so $x_i$ and $x_j$ are in $R$. Entities
  $x_i$ and $x_j$ cannot be at the same point at time $s_i$ since
  $x_i$ is the only entity on point $a_i$. The same holds for $s_j$. So they
  must have been together at $t_i$ and $t_j$. In particular, they must have
  been at points $c_i$ and $c_j$, and hence $v_i$ and $v_j$ are neighbours.

  The proof for the other direction, i.e., if $R'$ is a clique in $G$ then $R$
  is a robust group, is symmetrical. Clearly, the reduction is
  polynomial. Since it is also easy to check that a given set of entities forms
  a robust group we conclude that the problem is NP-complete.
\end{proof}


\end{document}